\documentclass[a4paper,oneside,openright]{article}
\usepackage{epsf}
 \usepackage{epsfig}
\usepackage{srcltx}
\usepackage{jcappub} 
\usepackage{mathrsfs}
\usepackage{amsmath}
\usepackage{amssymb}
\usepackage{graphicx}% Include figure files
\usepackage{dcolumn}% Align table columns on decimal point
\usepackage{bm}% bold math
\usepackage{latexsym}
\usepackage{amssymb}
\usepackage{amsmath}
\usepackage{hyperref}
\usepackage{amsthm}
\usepackage{multirow}
\usepackage{color}
\newtheorem{prop}{Proposition}

\newcommand{\be}{\begin{equation}}
\newcommand{\ee}{\end{equation}} %\indent}
\newcommand{\eei}{\end{equation}\indent\indent}
\newcommand{\bc}{\begin{center}}
\newcommand{\ec}{\end{center}}
\newcommand{\ber}{\begin{eqnarray*}}
\newcommand{\ear}{\end{eqnarray*}}
\newcommand{\ba}{\begin{array}}
\newcommand{\ea}{\end{array}}
 
\newcommand{\na}{\nabla}

\newcommand{\ti}{\tilde}

\newcommand{\bea}{\begin{eqnarray}}
\newcommand{\eea}{\end{eqnarray}}
\newcommand{\nn}{\nonumber}

\newcommand{\ei}{\end{itemize}}

\newcommand{\bra}[1]{\left(#1\right)}
\newcommand{\bras}[1]{\left[#1\right]}

\newcommand \veps {\varepsilon} %curly epsilon

\newcommand{\Lietwo}{{\cal L}}

\newcommand{\reff}[1]{(\ref{#1})}

\def\case#1/#2{\textstyle\frac{#1}{#2} }

\def\fp{f^{\prime}}
\def\fpp{f^{\prime \prime}}

\title{\boldmath Scalar wave scattering from Schwarzschild black holes in modified gravity}

\author[a]{Dan B. Sibandze,}
\author[a]{Rituparno Goswami,}
\author[a]{Sunil D. Maharaj,}
\author[a]{Anne Marie Nzioki,}
\author[b]{Peter K. S. Dunsby}

\affiliation[a]{Astrophysics \& Cosmology Research Unit, 
School of Mathematics Statistics and Computer Science,
University of KwaZulu-Natal, 
Private Bag X54001, Durban 4000, South Africa.}
\affiliation[b]{Department of Mathematics and Applied Mathematics and ACGC, University of Cape Town, Cape Town,7701,  South Africa.}

\emailAdd{danx36@gmail.com}
\emailAdd{Goswami@ukzn.ac.za}
\emailAdd{maharaj@ukzn.ac.za}
\emailAdd{anne.nzioki@gmail.com}
\emailAdd{peter.dunsby@uct.ac.za}

\abstract{
We consider the scattering of gravitational waves off a Schwarzschild black hole in $f(R)$ gravity. We show that the reflection and transmission coefficients for tensor waves are the same as in General Relativity. While the scalar waves, which are not present in General Relativity, demonstrate interesting features. The equation that governs these scalar waves can be reduced to a Volterra integral equation. Analysis of this equation shows that  a larger fraction of these waves are reflected compared to what one obtains for tensors. This may provide a novel observational signature for fourth order gravity. }
\begin{document}
\sloppy
\maketitle
\nopagebreak
%%%%%%%%%%%%%%%%%%%%%%%%%%%%%%%%%%%%%%%%%%%%%%%%%%%%%
\section{Introduction} 
%%%%%%%%%%%%%%%%%%%%%%%%%%%%%%%%%%%%%%%%%%%%%%%%%%%%%
Over the past hundred years General Relativity (GR) has matured into what is now arguably one of the most successful theories of modern physics. It has allowed us to explain gravitational phenomena from solar system scales \cite{Clifton, Hu, Capozziello, Guo, Berry} all the way to some of the largest scales in the observable universe. With the first two direct detections of gravitational waves from coalescing black holes by LIGO \cite{ligo1,ligo2}, the past year has been a particularly triumphant period for GR.

Despite these successes, most well established tests of GR still only involve weak gravitational fields and motions with speeds much less that the speed of light.  While the recent LIGO events represented the first real strong-field tests 
of the theory and were consistent with GR, many more such observations will me needed to probe the dynamical 
features of the strong field regime, before we can be certain that all extensions of Einstein gravity can be ruled out. Furthermore, there remain issues on cosmological scales, where it appears that the standard model based on GR plus a cosmological constant (or dynamical Dark Energy), together with dark matter might be either incomplete or suffer from extreme fine-tuning resulting from the coincidence and cosmological constant problems. Indeed, one of the main motivations for considering alternatives to GR arise from the rather obscure nature of Dark Energy and Dark Matter candidates. In fact,  if we wish to retain the Robertson-Walker metric for describing the large-scale geometry of the universe, it appears that one of the only physically acceptable alternatives to a Dark Energy driven late-time acceleration is one where the acceleration is a consequence of the breakdown of GR on cosmological scales. This mismatch between astrophysical and cosmological scales is arguably one of the biggest problems facing fundamental physics today.

Some of the most natural and promising extensions to GR are those which appear as the low energy limit of fundamental theories such as String or M-theory (e.g., \cite{Strings}). Examples of such modifications of GR can be found in a particularly popular and now very extensively studied class of fourth order theories of gravity, the so called $f(R)$ theories of gravity. In these theories, the modification to the gravitational action is described by the addition of a general function of the Ricci scalar $R$ which leads to field equations which are fourth-order in the metric tensor $g_{ab}$ (in GR the field equations are second order in $g_{ab}$).. This implies that the gravitational generated by the usual spin-2 graviton degrees of freedom together with a scalar degree of freedom. These deviations from GR derive from the work on scalar-tensor theory by Brans and Dicke, Jordan and Fierz. \cite{BD,Jordon,Fierz}. 

On cosmological scales, we require that $f(R)$ theories reproduce cosmological dynamics consistent with type Ia supernovae, BAO, Large Scale Structure and CMB measurements. They should be free from tachyonic instabilities, sudden singularities and ghosts and they should have valid Newtonian and post-Newtonian limits  \cite{Sulona}. We should also expect that well defined solutions found in GR, such as the Schwarzschild solution, are stable against generic perturbations in this more general context. Failure to satisfy the aforementioned criteria disfavours the theory as a viable alternatives to GR.

In GR, linear perturbations of Schwarzschild black holes were first studied in detail by Chandrasekhar using the metric approach together with the Newman-Penrose formalism \cite{chandra}. More recently, the standard results of Black Hole perturbation theory were reproduced using the 1+1+2 covariant approach \cite{Clarkson}. In the metric approach, perturbations are described by two wave equations, i.e., the Regge-Wheeler equation for odd parity modes and the Zerilli equation in the even parity case. These wave equations are described by functions (and their derivatives) in the perturbed metric which are not gauge-invariant, as general coordinate transformations do not preserve the form of the wave equation. However, using the 1+1+2 covariant approach, Clarkson and Barrett \cite{Clarkson} demonstrated that both the odd and even parity perturbations may be unified in a single covariant wave equation, which is equivalent to the Regge-Wheeler equation. This wave equation is governed by a single covariant, gauge and frame-independent, transverse-traceless tensor. These results were extended to include couplings (at second order) to a homogeneous magnetic field leading to an accompanying electromagnetic signal alongside the standard tensor (gravitational wave modes) \cite{betschart} and to electromagnetic perturbations on general locally rotationally symmetric spacetimes \cite{Burston}. The 1+1+2 covariant approach was later applied to $f(R)$ gravity in \cite{Anne,Pratten} where all calculations were performed in the Jordan frame. The dynamics of the extra gravitational degree of freedom inherent in these fourth order theories was determined by the trace of the effective Einstein equations, leading to a linearised scalar wave equation for the Ricci scalar. 

Of particular importance is the broader validity of the Jebsen-Birkhoff theorem, which in GR GR states that any the spherically symmetric spacetime of the field equations has to be necessarily static or spatially homogenous. 

 In GR, it was found that this theorem was ``stable'' even with the introduction of small perturbations. More precisely,  almost spherical symmetry and/or almost vacuum imply almost static or almost spatially homogeneous \cite{amostBirk,amostBirkmatter,varBirk}. However, in higher order theories of gravity, such as fourth order theories of gravity, the validity of this theorem depends on extra conditions being satisfied. Some progress in this area was made in \cite{BirkfR}, where it was found that a non-zero measure exists within the parameter space of $f(R)$ theories for which stability holds under generic perturbations for a Jebsen-Birkhoff like theorem.

In this paper we consider some further perturbative results related to Schwarzschild black holes in $f(R)$ gravity. We investigate in detail how the scalar waves from infinity scatter off black holes due to the one dimensional potential barrier of the ``Schrodinger-like'' equation which governs the perturbations. This is done by computing the reflection and transmission coefficients and comparing them with what is found in GR. We find that for short wavelengths, a larger fraction of the spin-0 waves get reflected from the black hole potential barrier (in comparison to the spin-2 tensor waves). This may provide a novel observational signature for modified gravity.

Unless otherwise specified, geometric units ($8\pi G=c=1$) will be used throughout this paper. 

%%%%%%%%%%%%%%%%%%%%%%%%%%%%%%%%%%%%%%%%%%%%%%%%
\section{Higher Order Gravity}
%%%%%%%%%%%%%%%%%%%%%%%%%%%%%%%%%%%%%%%%%%%%%%%%

In general relativity (GR) the Einstein-Hilbert action is given as
\be
{\cal S} =\frac12 \int dV \bras{ \sqrt{-g} \bra{R-2\Lambda} +2\, \Lietwo_{M}(g_{ab}, \psi) } , 
\label{EHaction}
\ee
where $ \Lietwo_{M} $ is the Lagrangian density of the matter fields $\psi$, $ R $ is the Ricci 
scalar and $ \Lambda $ is the cosmological constant. The invariant 4-volume element 
is given by the expression $\sqrt{-g}\,dV$ and the gravitational Lagrangian density 
as $ \Lietwo_{g} = \sqrt{-g} \bra{R-2\Lambda}$, where $g$ is the determinant 
of the metric tensor $g_{ab}$. A generalisation of this 
action is done by replacing $R$ in (\ref{EHaction}) with a $C^2$ function of the 
quadratic contractions of the Riemann curvature tensor $R^{2}$, $R_{ab}R^{ab}$, 
$R_{abcd}R^{abcd}$ and $\veps^{klmn} R_{klst}R^{st}{}_{mn}$ where $\veps^{klmn}$ 
is the antisymmetric 4-volume element. In fact, in the quantum field picture, the effects 
of renormalisation are expected to add such terms to the Lagrangian in order to give a 
first approximation to some quantised theory of gravity \cite{DeWitt:1967,Birrell+Davies}. 
The Lagrangian density that can be constructed from the generalisation is of the form
\be
\Lietwo_{g}=\ \sqrt{-g} ~f(R, R_{ab}\,R^{ab}, R_{abcd}\,R^{abcd} ) ~.
\ee
It is a well known result that \cite{DeWitt:1965, Buchdahl:1970, Barth:1983}, 
\begin{align}
(\delta/\delta g_{ab})\int &  dV \bra{R_{abcd}\,R^{abcd}-4R_{ab}\,R^{ab}+R^2}=0\;,\\
(\delta/\delta g_{ab})\int & dV \,\veps^{klmn} R_{klst}\,R^{st}{}_{mn}= 0\;,
\end{align}
that is, the functional derivative of the Gauss-Bonnet invariant 
$R_{abcd}\,R^{abcd}-4R_{ab}\,R^{ab}+R^2$ 
and $\epsilon^{iklm}R_{ikst}\,R^{st}{}_{lm}$ vanish with respect to $g_{ab}$.
If we consider the function $f$ to be linear in $R_{abcd}R^{abcd}$, we can 
use this symmetry to rewrite $R_{abcd}R^{abcd}$ in terms of the other two invariants and 
as a result the action for FOG can be written as:
\bea
{\cal S} &=& \frac12 \int dV \left\{ \sqrt{-g} \bra{c_{0}\,R+ c_{1}\,R^{2}
+ c_{2}\,R_{ab}\,R^{ab}}\right. \nonumber\\&&\left.+2\,\Lietwo_{M}(g_{ab}, \psi) \right\}~. 
\label{generalfr}
\eea
where the coefficients $c_{0}$, $c_{1}$ and $c_{2}$ have the appropriate dimensions. 
Similarly, if the spacetime is homogeneous and isotropic, then because of the 
following identity,
\be
(\delta/\delta g_{ab})\int  dV\,\bra{3R_{ab}\,R^{ab}-R^2}=0\;,
\ee
the term $R_{ab}\,R^{ab}$ can always be rewritten in terms of the variation of $R^2$. 
Though in the present paper we are not discussing isotropic spacetimes, nevertheless even for the spherically symmetric case we can safely assert that  a sufficiently general and ``effective'' fourth-order Lagrangian for  highly symmetric spacetimes contain only powers of $R$. Also this makes the problems more physically realistic as it has been shown that the theories that contain the square of Ricci tensor in the action, suffer from several instabilities. 

Therefore we can write the action as 
\be
{\cal S}= \frac12 \int dV \bras{\sqrt{-g}\,f(R)+ 2\,\Lietwo_{M}(g_{ab}, \psi) } ~.
\label{action}
\ee 
This action represents the simplest generalisation of the Einstein-Hilbert density. 
Demanding that the action be invariant under some symmetry ensures that the 
resulting field equations also respect that symmetry. That being the case, since the 
Lagrangian is a function $R$ only, and $R$ is a generally covariant and locally Lorentz 
invariant scalar quantity, then the field equations derived from the action (\ref{action}) are 
generally covariant and Lorentz invariant. \\
\\
There are different variational principles that can be applied to the action $\cal S$
in order to obtain the field equations. One approach is the \textit{standard metric 
formalism} where variation of the action is with respect to the metric $g_{ab}$ and 
the connection $\Gamma^a_{~~bc}$ in this case is the Levi-Civita one, that is, the metric 
connection
\be
\Gamma^a{}_{bc}=\frac{1}{2}\,g^{ad} \bra{g_{bd,c}+g_{dc,b}-g_{bc,d}} ~.
\ee
%%%%%%%%%%%%%%%%%%%%%%%%%%%%%%%%%%%%%%%%%%%%
\section{Field equations in metric formalism}
%%%%%%%%%%%%%%%%%%%%%%%%%%%%%%%%%%%%%%%%%%%%

Varying the action \reff{action} with respect to the metric $g_{ab}$ over a 4-volume 
yields:
\bea
\delta{\cal S} &=& -\frac12 \int dV  \, \sqrt{-g} \left\{\frac12 f \, g_{ab} \, \delta g^{ab} 
-\fp \,\delta R +T^{M}_{ab} \, \delta g^{ab}\right\} ~,
\eea
where $ ' $ denotes differentiation with respect to $R$, and $ T^{M}_{ab} $ is the matter 
\textit{energy momentum tensor} (EMT) defined as 
\be
T^{M}_{ab}=- \frac{2}{\sqrt{-g}} \, \frac{\delta \Lietwo_{M}} {\delta g^{ab} } ~.
 \label{metricEMT}
\ee
Writing the Ricci scalar as $R= g^{ab}\,R_{ab}$ and assuming the connection is the 
Levi-Civita one, we can write
\be
 \fp  \,\delta R \simeq  \delta g^{ab}\bra{\fp \, R_{ab}
 + g_{ab} \, \Box  \fp- \na_{a}\na_{b} \fp}~,
\ee
where the $\simeq$ sign denotes equality up to surface terms and 
$\Box \equiv \na_{c}\na^{c}$. By requiring that $\delta{\cal S} =0$ with respect to variations in the metric, ergo a stationary action, one has finally 
\bea
 \fp \bra{R_{ab}-\frac12 g_{ab} \, R}&= & \frac12 g_{ab}\,(f-R \,  \fp) 
 +\na_{a}\na_{b}\fp \nonumber\\&&- g_{ab} \,\Box  \fp \label{field1}+ T^{M}_{ab}~.
\eea
The special case $f = R$ gives the standard Einstein field equations.\\
\\
It is convenient to write \reff{field1} in the form of effective Einstein equations as
\be
G_{ab} = \bra{R_{ab}-\frac12 g_{ab} \, R} 
= \tilde{T}^{M}_{ab} + T^{R}_{ab} = T_{ab}~,
\label{field2} 
\ee
where we define $T_{ab}$ as the total EMT with
\be
\tilde{T}^{M}_{ab} = \frac{T^{M}_{ab}}{ \fp} \,,
\label{matterEMT} 
\ee 
and
\be
T^{R}_{ab} = \frac{1}{ \fp} \bras{\frac12g_{ab} \,(f-R\, \fp) 
+\na_{a}\na_{b} \fp- g_{ab}\,\Box \fp}~. 
\label{curvatureEMT}
\ee
The field equations \reff{field2} contain fourth order derivatives of the metric functions, which 
can be seen from the existence of the $\na_{a}\na_{b} \fp$ term in 
\reff{curvatureEMT}. This result also follows from a corollary of 
Lovelock's theorem \cite{Lovelock:1971, Lovelock:1972} which states that in a 
four-dimensional Riemannian manifold, the construction of a metric theory of modified 
gravity must admit higher than second order derivatives in the 
field equations. Though this is an undesirable feature in a Lagrangian 
based theory as it can lead to Ostrogradski instabilities \cite{Ostrogradsky:1850} 
in the solutions of the field equations, the $f(R)$ theories are special as in these this instability can be avoided \cite{Woodard:2007}, due to the existence of an 
equivalence with scalar-tensor theories. \\

%%%%%%%%%%%%%%%%%%%%%%%%%%%%%%%%%%%%%%%%%%%%
\section{Schwarzschild solution and it's stability}
%%%%%%%%%%%%%%%%%%%%%%%%%%%%%%%%%%%%%%%%%%%%

We know that in GR, the rigidity of spherically symmetric vacuum solutions of Einstein's field equations continues 
even in the perturbed case. Particularly, almost spherical symmetry and/or almost 
vacuum implies almost static or almost spatially homogeneous \cite{amostBirk,amostBirkmatter,varBirk}. 
This result emphasises the stability of Schwarzschild solution in general relativity. 

In $f(R)$-gravity, the extension of this result is not so obvious due to the presence of an extra scalar degree of freedom in the field equations. However, it has been shown recently that a Birkhoff-like theorem does exist in these theories \cite{BirkfR}, that states the following:
\textit{For $f(R)$ gravity, where the function $f$ is of class $C^3$ at $R=0$, with $f(0) = 0$ 
and $f'_{0} \ne 0$, the only spherically symmetric solution with vanishing Ricci scalar in 
empty space in an open set $\mathcal{S}$, is one that is locally equivalent to part of 
maximally extended Schwarzschild solution in $\mathcal{S}$.} The stability of this local theorem in the perturbed case has been formulated as: : {\em For $f(R)$ gravity, where the function $f$ is of class $C^3$ at $R=0$, with $f(0) = 0$ 
and $f'_{0} \ne 0$, any almost spherically symmetric solution with almost vanishing Ricci scalar in 
empty space in an open set $\mathcal{S}$, is locally almost equivalent to part of 
maximally extended Schwarzschild solution in $\mathcal{S}$.} The important point to note here is that the size of the open set ${\mathcal S}$ depends on the 
parameters of the theory (namely the quantity $\fpp(0)$) and the Schwarzschild mass) and they can be always tuned such that 
the perturbations continue to remain small for a time period which is greater than the age of the 
universe. This clearly indicates that the local spacetime around almost spherical stars will be stable in the 
regime of linear perturbations in these modified gravity theories.

%-------------------------------------------------------------------------------------------------------------------------------------------
\subsection{Linear perturbation of Schwarzschild black hole in f(R) gravity}
%--------------------------------------------------------------------------------------

In GR, the two fundamental second-order wave equations govern the gravitational perturbations of Schwarzschild black holes are the Regge-Wheeler equation \cite{Regge} and the Zerilli equation \cite{Zerilli}. The former equation describes the the odd perturbations and the latter the even perturbations. Both the equations satisfy a Schrodinger-like equation and the effective potentials of these equations is shown to have the same spectra \cite{Chandrasekhar}. These waves are tensorial, and are sourced by small deviation from the spherical symmetry of the Schwarzschild black hole in vacuum. 

For $f(R)$-gravity it is evident from the {\it Almost Birkhoff-like theorem} stated in the previous section that there can be two types of perturbations. The first is the tensor perturbation driven by small departure from the spherical symmetry (like GR), whereas the second one is the scalar perturbation that is sourced by perturbations in the Ricci scalar, which vanishes in the unperturbed background. This is an extra mode, that is generated by the extra scalar degree of freedom in these theories and is absent in GR. The detections of these modes are of a crucial importance in asserting the validity or otherwise of GR as the theory of gravity. We will now briefly discuss about the wave-equations governing these two different kind of perturbations in $f(R)$-gravity.

%---------------------------------------------------------------------------------------
\subsubsection{Tensor perturbations}
%--------------------------------------------------------------------------------------

In \cite{Anne}, it has been explained in detail, that in $f(R)$-gravity, one can construct a transverse traceless gauge independent 2-tensor, whose coefficients of harmonic decomposition $M_T$ obey the same Regge-Wheeler equation as in GR. In terms of the `tortoise' coordinate $r_*$, which is related to the usual radial coordinate  $r$ by
\be
r_*=r+2m\,\ln\bra{\frac{r}{2m}-1}~,
\label{tortoise}
\ee
this equation can be written in the form
\be
\bra{\frac{d^2}{dr_*^2}+\kappa^2-V_T}M_T= 0~,
\label{schroedTens}
\ee
with the effective potential $V_T$  
\be
V_T=\bra{1-\frac{2m}{r}}\bras{\frac{\ell\bra{\ell+1}}{r^2}
-\frac{6m}{r^3}}~,
\label{Gravitypot}
\ee 
and we have factored out the harmonic time dependence part of $M_T$, which is $\exp(i\kappa t)$.
As $V_T$ is the \emph{Regge\,-Wheeler potential} for 
\emph{gravitational perturbations}. This clearly indicates that the tensorial modes of the gravitational perturbations in $f(R)$-gravity have the same spectrum as in GR and hence observationally it is impossible to differentiate between the two through these modes.

%---------------------------------------------------------------------------------------
\subsubsection{Scalar perturbations}
%--------------------------------------------------------------------------------------

Taking the trace of the equation (\ref{field2}) in vacuum we get 
\be\label{trace}
3\Box \fp +R\fp -2f=0\;,
\ee
 which is a wave equation in terms of the Ricci scalar $R$ associated with scalar modes. These modes are not present in GR as can be seen by substituting $f(R)=R$ in the above equation, which gives $R=0$. Hence in vacuum spacetimes in GR there can not be any perturbations in Ricci scalar. However this is possible in $f(R)$ gravity and we can Taylor expand the function $f$ around $R=0$ (using $f(0)=0$ for the existence of Schwarzschild solution) to get
\be\label{expansion}
f(R) = f'_0R + \frac{f''_0}{2}R + \hdots\;.
\ee
Using the tortoise coordinates, rescaling $R=r^{-1}\mathcal{R}$, and factoring out the time dependence part $\exp(i\kappa t)$ from $\mathcal{R}$ we get,
\begin{equation}
\left(\frac{d^2}{dr^{2}_{*}}+\kappa^{2}-V_{S}\right)\mathcal{R}=0 \label{ode1}
\end{equation}
where
\begin{equation}\label{Scalpot}
V_{S}=\left(1-\frac{2m}{r}\right)\left[\frac{l(l+1)}{r^2}+\frac{2m}{r^3}+{U}^2\right]\;
\end{equation}
is the Regge-Wheeler potential for the scalar perturbations and 
\be
{U}^{2}=\frac{f'_{0}}{3f''_{0}}.
\ee
The form of the wave equations \reff{ode1} is similar to a one dimensional 
Schr\"{o}dinger equation and hence the potential correspond to a single 
potential barrier. This equation can be made dimensionless by multiplying through with 
the squire of the black hole mass $m$. In this way the potential \reff{Scalpot} 
becomes
\bea
V_S&=&\bra{1-\frac{2}{r}}\bras{\frac{\ell\bra{\ell+1}}{r^2}
+\frac{2}{r^3}+u^2}~,
\label{normalpotential}
\eea
where we have defined (and dropped the tildes),
\be
\tilde{r} = \frac{r}{m}~, \qquad \ti{u} = m\,{U}~, \qquad \ti{\kappa}=m\kappa \;.
\label{dimensionless}
\ee
\begin{figure}[ht] 
\caption{Potential profile $V_S$ for l=2,  3, 4}
   \centering
   \includegraphics[width=8cm]{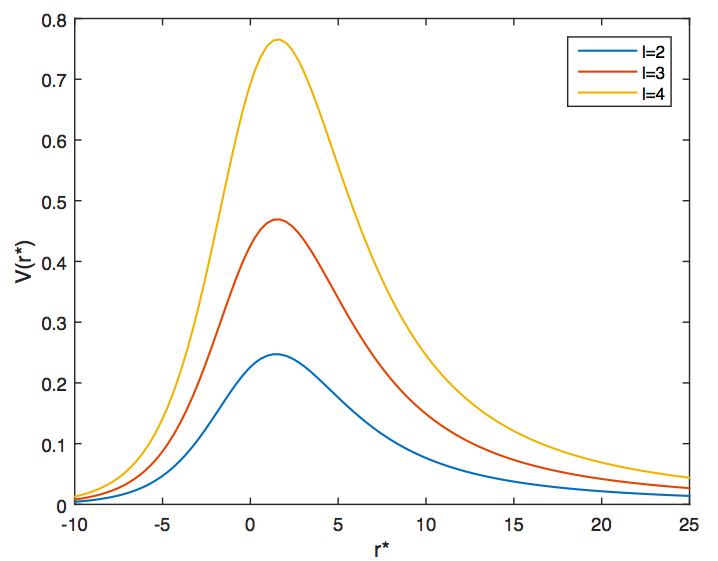} 
\end{figure}
For scalar perturbations with $u = 0$, the potential has two extrema, one in the unphysical region $r<0$ and the other in  $r>0$. In the case of the scalar perturbations with $u \ne 0$, for a certain range of $u$, the potential has three extrema: one in the unphysical region $r<0$, a local maximum at $r_{max}$ and local minimum at  $r_{min}$ such that $2<r_{max}<r_{min}$. 

\section{Infra-red Cutoff for incoming Scalar Waves}

Let us now look at the equation governing the scalar waves (\ref{ode1}) and the form of the potential (\ref{normalpotential}), to study the limiting behaviour of the waves. This will help us specify the physically realistic boundary conditions. At $r_{*}\rightarrow -\infty$, (which implies the horizon at $r=2$), we have $V_{S}=0$, and equation (\ref{ode1}) becomes
\begin{equation}
\left(\frac{d^2}{dr^{2}_{*}}+\kappa^{2}\right)\mathcal{R}=0\;, \label{ode2}
\end{equation}
which is an usual harmonic equation with two linearly independent solutions
\be\label{sol1}
\mathcal{R}\sim C_1 \exp{(i\kappa r_{*})} + C_2 \exp{(-i\kappa r_{*})}\;.
\ee
Since we do not have any {\it outgoing} mode at the horizon, this implies $C_2=0$.  On the other hand,
at $r_{*} = +\infty$, equation (\ref{ode1}) becomes
\begin{equation}
\left(\frac{d^2}{dr^{2}_{*}}+\kappa^{2}-u^2\right)\mathcal{R}=0 \label{ode2}
\end{equation}
with
  \be
 \mathcal{R}\sim C_3\exp{(i\sqrt{\kappa^2-u^2}r_*)} +C_4\exp{(-i\sqrt{\kappa^2-u^2}r_*)}\;.
  \ee
At this point, we come to a very important proposition which we state as follows:
\begin{prop}
The parameters of the theory in $f(R)$ gravity provides a cut-off for long wavelength spherical incoming waves from infinity.
\end{prop}
\begin{proof}
 When $u^2 > \kappa^2$, we can immediately see for the incoming modes, 
 \be
 \lim_{r_* \to \infty} \mathcal{R}_{in} =  C_3\exp{(-\sqrt{-\kappa^2+u^2}r_*)}\to 0
 \ee
 Hence, there are no incoming scalar waves at $r_* \to \infty$ for $\kappa<u.$
  \end{proof}
 As we are interested in the scattering of incoming scalar waves from infinity by the black hole potential barrier, in the following sections we choose the parameters of the theory, such that $u^2<< \kappa^2$. Hence for all practical purposes we have ${\kappa'}\equiv\sqrt{\kappa^2 - u^2}=\kappa$.
 
 %---------------------------------------------------------------------------------------
\section{Study of potential scattering using Jost functions}
%--------------------------------------------------------------------------------------

In this section we investigate in detail, how the scalar waves from infinity get scattered by the black holes in $f(R)$-gravity. This scattering (that depicts the reflexion and transmission) is due to the one dimensional potential barrier of the Schrodinger-like equation governing the perturbations. We set our boundary conditions in a way that there is no {\it outgoing} wave from the event horizon. Considering an influx of incoming waves from infinity, we would like to know that what fraction of these waves gets reflected by the potential barrier and what fraction gets transmitted to the black hole. Our analysis here is quite similar to the analysis presented in \cite{chandra}. We use the method of {\it Jost function}, which is the Wronskian of the regular solution and the (irregular) Jost solution to the differential equation.

Equation ($\ref{ode1}$) is an ODE integrable over $(-\infty, \infty)$. Moreover, $V_{S}(-\infty) = 0$ and $V_{S}(\infty) = u^2$. If we let $r_{*}\rightarrow \pm \infty$ in equation $\ref{ode1}$, we obtain two particular solutions with the asymptotic behaviours
$$\mathcal{R}_{1}(r_{*}, \kappa)\sim e^{-i\kappa' r_{*}}\sim e^{-i\kappa r_{*}}, \qquad (r_{*} \rightarrow +\infty)$$ and 
$$\mathcal{R}_{2}(r_{*}, \kappa) \sim e^{i\kappa r_{*}}, \qquad (r_{*} \rightarrow -\infty)$$ 
which are independent since their Wronskian
\begin{align}
\left[ \mathcal{R}_{1}(r_{*}, \kappa), \mathcal{R}_{2}(r_{*}, \kappa)\right] &= +2i\kappa \neq 0.
\end{align}

For real $\kappa$, the solution represents ingoing and outgoing waves at $\pm \infty$. This problem becomes one of reflection and transmission of incident waves by the potential barrier, $V_{S}$.  We seek solutions satisfying of the wave equation (\ref{ode1}) and the boundary conditions, 

\begin{equation} \label{eqn02}
\mathcal{R}_{2}(r_{*}, \kappa) = \frac{R_{1}(\kappa)}{T_1(\kappa)} \mathcal{R}_1(r_{*}, \kappa) ) + \frac{1}{T_1(\kappa)} \mathcal{R}_{1}(r_{*}, -\kappa) 
\end{equation}
and 
\begin{equation} \label{eqn03}
\mathcal{R}_{1}(r_{*}, \kappa) = \frac{R_2(\kappa)}{T_2(\kappa)} \mathcal{R}_2(r_{*}, \kappa) ) + \frac{1}{T_2(\kappa)} \mathcal{R}_{2}(r_{*}, -\kappa)
\end{equation}
where $R_1(\kappa)$, $R_2(\kappa)$, $T_1(\kappa)$, $T_2(\kappa)$ are distinct functions that exist if $\kappa \neq 0$.
Here we can easily see that $T_1(\kappa)\mathcal{R}_{2}(r_{*}, \kappa) $ corresponds to an $incident$ $wave$ of unit amplitude from $+\infty$ giving rise to a $reflected$ $wave$ of amplitude ${R}_1(\kappa)$ and a transmitted wave of amplitude $T_1(\kappa)$.  
In the theory of potential scattering, the Jost functions are defined by
\begin{equation}
 m_{1}(r_{*}, \kappa) = e^{+ i \kappa r_{*}} \mathcal{R}_{1}(r_{*}, \kappa) 
 \end{equation}
 and
 \begin{equation}
 m_{2}(r_{*}, \kappa) = e^{- i \kappa r_{*}} \mathcal{R}_{2}(r_{*}, \kappa)
\end{equation}
which satisfy the boundary conditions
\begin{align}
m_{1}(r_{*}, \kappa) &\to 1 \quad \text{as} \quad r_{*} \to +\infty \nn \\
\text{and} \quad m_{2}(r_{*}, \kappa) &\to 1 \quad \text{as} \quad r_{*} \to -\infty. \label{bc2}
\end{align}

Equations (\ref{eqn02}) and (\ref{eqn03}) can be respectively written in terms of the Jost functions as
\begin{align}
T(\kappa)  m_{2}(r_{*}, \kappa) &= R_{1}(\kappa) e^{-2 i \kappa r_{*}} m_{1}(r_{*}, \kappa) + m_{1}(r_{*}, -\kappa),
\end{align}
and 
\begin{equation}
T(\kappa) m_{1}(r_{*}, \kappa) = R_{2}(\kappa) e^{+2 i \kappa r_{*}} m_{2}(r_{*}, \kappa) + m_{2}(r_{*}, -\kappa),
\end{equation}
where $T_1(\kappa) = T_2(\kappa)= T(\kappa)$.
From the conditions imposed in (\ref{bc2}), it follows that 
\begin{equation}
m_{1}(r_{*}, \kappa) = \frac{R_{2}(\kappa)}{T(\kappa)} e^{+2 i \kappa r_{*}} + \frac{1}{T(\kappa)} + o(1)  \quad (r_{*} \to -\infty),
\end{equation}
and
\begin{equation}
 m_{2}(r_{*}, \kappa) = \frac{R_{1}(\kappa)}{T(\kappa)} e^{-2 i \kappa r_{*}} + \frac{1}{T(\kappa)} + o(1) \quad (r_{*} \to +\infty) .\label{eqn49}
\end{equation}
Let us now write 
\be
\mathcal{R}_{2}(r_{*}, \kappa) = e^{i \kappa r_{*}} + \psi(r_{*}, \kappa).
\ee
We note that $\psi \to 0$ as $r_{*} \to -\infty$, and  $\psi$ satisfies the differential equation
\begin{equation}
\left(\frac{d^{2}}{d r_{*}^{2}} + \kappa^{2}\right) \psi = \left(e^{i \kappa r_{*}} + \psi \right) V_{s}.
\end{equation}

Now we know that, given any linear ODE of the form
$L \psi(x) = -f(x)$, where $L$ is the linear harmonic differential operator, 
the solution is given by Green\rq{}s function
\be
\psi(x) = \int G(x, x\rq{})f(x\rq{}) dx\rq{},
\ee
where 
\be
G(x, x\rq{}) = \frac{1}{\kappa} \left[ \frac{1}{2 i} \left( e^{i \kappa (x-x\rq{})} - e^{-i \kappa (x-x\rq{})}\right) \right]\;.
\ee
Therefore we can write the solution $\psi(x)$ in the form:
\begin{align}
\psi(r _{*}, \kappa) =& \frac{1}{2 i \kappa} \int_{-\infty}^{r_{*}}\left[ e^{i \kappa (r_{*}\rq{} - r_{*})} - e^{-i \kappa (r_{*}\rq{} - r_{*})} \right] V_{S}(r_{*}\rq{}) \nn \\
&\times \left[ e^{i \kappa r_{*}} + \psi(r_{*}, \kappa) \right] dr_{*}\rq{}
\end{align}
Using the above equations we now get an integral equation for the Jost function as 
 \begin{align}\label{veqn}
m_{2}(r _{*}, \kappa) = &1 - \frac{1}{2 i \kappa} \int_{-\infty}^{r _{*}} \left( e^{2 i \kappa (r_{*}\rq{} - r_{*})} - 1 \right) \nn \\
&\times V_{s}(r_{*}\rq{}) m_{2}(r_{*}\rq{}, \kappa) dr_{*}\rq{}
\end{align}
which is a Volterra integral equation of the second kind. In the next section we give a numerical scheme to solve this equation, which will then provide us the required expressions for reflected and transmitted waves. In figure 1, we have plotted the nature of the Jost function $m_2(r_{*}\rq{}, \kappa)$.

 \begin{figure}[ht] 
\caption{Jost function for l=2; u=0.001}
   \centering
   \includegraphics[width=7cm]{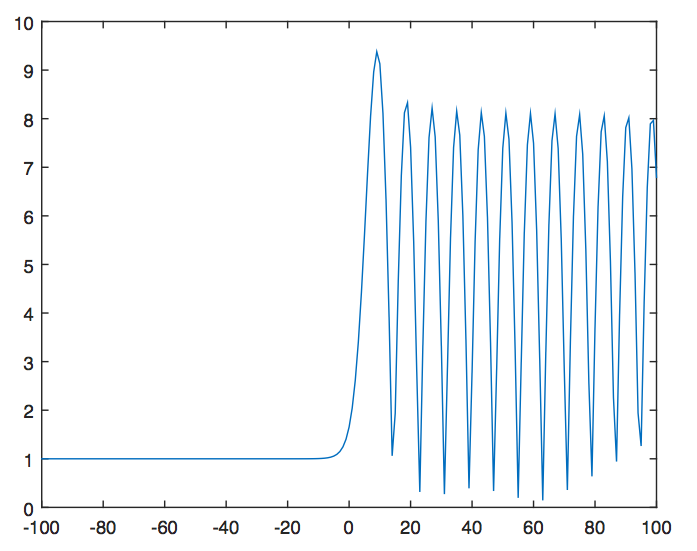}    %u3=>u=0.001
\end{figure}

%---------------------------------------------------------------------------------------
\section{Numerical Solution}
%---------------------------------------------------------------------------------------
 Given a  Volterra Integral Equation of the second kind (\ref{veqn}), which is of the form
 \begin{equation} \label{vie}
 u (x) = f(x) + \lambda \int_a^x K(x,y)u(t)dt,
 \end{equation}
 we divide the interval of integration ($a,x$) into $n$ equal subintervals, $\Delta t = \frac{x_n - a}{n}$, where $n \geq 1$ and $x_n = n.$
 Also let $y_0 = a$, $x_0 = t_0$, $x_n = t_n = x$, $t_j = a + j\Delta t = t_0  + j\Delta t$,  $x_0 + i\Delta t = a + i\Delta t = t_i.$ 
 Using the trapezoid rule, the integral can now be written as
 \begin{align}
 &\int_{a}^x K(x,t) u(t) dt \nn \\
 &\approx \Delta t \left[\frac{1}{2}K(x,t_0) u(t_0) + K(x,t_1) u(t_1) + \hdots \right. \nn \\
 &+ \left. K(x,t_{n-1}) u(t_{(n-1)}) + \frac{1}{2} K(x,t_n) u(t_n) \right],
 \end{align}
 where $\Delta t = \frac{t_j - a}{j} = \frac{x-a}{n}, $ $t_j \leq x, $ $j \geq 1, $ $x = x_n = t_n.$\\
 Using the above, equation (\ref{vie}) can be discretised as
 \begin{align}  \label{discret}
 u(x)&= f(x) + \lambda \Delta t \left[\frac{1}{2}K(x,t_0) u(t_0) + K(x,t_1) u(t_1) + \hdots \right. \nn \\
 &+ \left. K(x,t_{n-1}) u(t_{(n-1)}) + \frac{1}{2} K(x,t_n) u(t_n) \right]~.
 \end{align}
 Since $K(x,t) \equiv 0$ when $t > x$ (the upper limit of the integration ends at $t = x$), then $K(x_i, t_j) = 0$ for $t_j > x_i$. Numerically, equation (\ref{discret}) becomes
 \begin{align}
 u(x_i) = f(x_i) + \lambda \Delta t \left[\frac{1}{2}K(x_i,t_0) u(t_0) + K(x_i, t_1) u(t_1) \right. \nn \\
 + \hdots + \left. K(x_i,t_{j-1}) u(t_{(j-1)}) + \frac{1}{2} K(x_i,t_j) u(t_j) \right] , 
 \end{align}
 where $i = 1, 2, \hdots ,n \quad t_j \leq x_i$ and $u(x_0) = f(x_0).$
 Denoting $u_i = u(x_i)$, $f_i = f(x_i)$ and $K_{ij} = K(x_i, t_j)$, we can write the numeric equation in a simpler form as 
 \begin{align}
 u_0 &= f_0 \nn \\
 u_i & = f_i + \lambda \Delta t \left[\frac{1}{2}K_{i0} u_0 + K_{i1} u_1 \right. + \hdots  + \left.K_{i(j-1)} u_{j-1} + \frac{1}{2} K_{ij} u_j \right] ,
 \end{align}
 with $i = 1, 2, \hdots ,n$ and $j \leq i$. Therefore there are $n+1$ linear equations 
 \begin{align}
 u_0 &= f_0 \nn \\
 u_1 &= f_1 + \lambda  \Delta t \left[\frac{1}{2}K_{10} u_0 + K_{11} u_1\right] \nn \\
 u_2 &= f_2 + \lambda \Delta t \left[\frac{1}{2}K_{20} u_0 + K_{21} u_1 + \frac{1}{2} K_{22}u_2\right] \nn \\
\vdots & \qquad \ \vdots \nn \\
 u_n &= f_n + \lambda \Delta t \left[\frac{1}{2}K_{n0} u_0 + K_{n1} u_1 + \hdots +K_{n(n-1)} u_{n-1} +  \frac{1}{2} K_{nn}u_n\right]\;.
 \end{align}
 Hence  a general equation can be written in compact form as
 \begin{equation}
 u_i  = \frac{f_i + \lambda \Delta t \left[\frac{1}{2}K_{i0} u_0 + K_{i1} u_1 + \hdots + K_{i(i-1)} u_{i-1} \right] }{1 - \frac{\lambda \Delta t}{2} K_{ii}}
 \end{equation}
 and can be evaluated by substituting $u_0, u_1, \hdots, u_{i-1}$ recursively from previous calculations. A MATLAB code was written to evaluate this system of linear equations  for (\ref{veqn}) and the results were used to evaluate the reflexion and transmission coefficients by coding the numerical solution for $m_{2}(x, \kappa)$ with the potential (\ref{normalpotential}) for different values of $u$.

\begin{table}
   \centering
   \begin{tabular}{@{}rrrrcrrrcrrr@{}} 
    \hline
&  \multicolumn{3}{c}{$\bm{R}$} \\
  \cline{2-4}  
      $\kappa$   & ${u=0}$ & ${ u=0.001}$ & ${ u=0.01}$\\
      \hline
   0.10 & 1.0000 & 1.0000 & 1.0000 \\
   0.20 & 0.9995 & 1.0000 & 1.000\\
   0.30 & 0.9690 & 0.9989 & 0.9991\\
   0.32 & 0.9382 & 0.9974 & 0.9980\\
   0.34 & 0.8837& 0.9946 & 0.9955\\
   0.36 & 0.7920 & 0.9886 & 0.9903\\
   0.40 & 0.5441 & 0.9698 & 0.9589\\
    0.50 & - & 0.5028 & 0.5028
   \end{tabular}
   \caption{The \emph{reflection amplitude} ($\bm{R}$), where $l=2$, for various frequencies ($\kappa$) and for different values of $u$.}
   \label{table1}
\end{table}

%--------------------------------------------------------------------------------------
\section{Scalar wave scattering: Results and discussions}
%--------------------------------------------------------------------------------------

It is well known \cite{chandra} that the solution to the Volterra integral equation (\ref{veqn}) is analytic for lower half of complex $\kappa$ plane and is continuous for $\Im(\kappa)\le0$.
In this case, the solution obtained by repeated iterations always converges and $m_{2}(r _{*}, \kappa)$ can be expanded as a power series in $1/\kappa$. These facts indicate the following:
\begin{align}
m_{2}(x, \kappa) &= 1 -e^{-2 i \kappa r_{*}} \frac{1}{2 i \kappa} \int_{-\infty}^{+\infty}  e^{2 i \kappa r_{*}\rq{}} V_{S}(r_{*}\rq{}) m_{2}(r_{*}\rq{}, \kappa) dr\rq{} + \nonumber\\
& \frac{1}{2 i \kappa} \int_{-\infty}^{+\infty}V_{S}(r_{*}\rq{}) m_{2}(r_{*}\rq{}, \kappa) dr_{*}\rq{} + o(1)\label{results}
\end{align}
 
Comparing the above result with equation (\ref{eqn49}) immediately gives the relation between reflexion and transmission coefficients and the Jost function as 
\begin{equation}
\frac{R_{1}(\kappa)}{T(\kappa)} = - \frac{1}{2 i \kappa} \int_{-\infty}^{+\infty}  e^{(2 i \kappa r_{*}\rq{})}  V_{S}(r_{*}\rq{}) m_{2}(r_{*}\rq{}, \kappa) dr_{*}\rq{}
\end{equation}
\begin{equation}
\frac{1}{T(\kappa)} = 1 +  \frac{1}{2 i \kappa} \int_{-\infty}^{+\infty}  V_{s}(r_{*}\rq{}) m_{2}(r_{*}\rq{}, \kappa) dr_{*}\rq{}
\end{equation}
From the above expression, the following conservation condition can be verified easily:
\be
\bm{R}+\bm{T}\equiv|R_1|^2+|T|^2=1
\ee

The reflection wave amplitude $\bm{R}$ for various frequencies and for different values of $u$ are summarised in table \ref{table1}. From this analysis we find a few interesting results, which we summarise as follows:
\begin{enumerate}
\item If we compare the reflection coefficients of the tensor waves in GR from \cite{chandra}, which will be the same in $f(R)$-gravity, we see that for low wavelengths, larger fraction of the scalar waves get reflected (in comparison to tensor waves) from the black hole potential barrier. This may provide a novel observational signature for modified gravity or otherwise. 
\item Furthermore from the table we can immediately see that as $u$ increases, the tendency of reflection increases for long wavelength scalar waves. This trait continues till the infra-red cutoff happens for a given frequency.
\item Also these calculations depict that as $u$ increases, reflection wave amplitude attains a plateau near $\bm{R}=1$ for long wavelengths that suddenly drops off for higher frequency.
\end{enumerate}

We would like to emphasise here that these results are only applicable in the scenario's where the frequency of the scalar waves are much larger than $u$ (which is given by the parameters of the theory of gravity considered). Otherwise there will be a completely different scenario in terms of localisation of the scalar waves, which will be reported elsewhere.

 \begin{acknowledgments}
All the authors are supported by National Research Foundation (NRF), South Africa. SDM 
acknowledges that this work is based on research supported by the South African Research Chair Initiative of the Department of
Science and Technology and the National Research Foundation.
\end{acknowledgments}

%%%%%%%%%%%%%%%%%%%%%%%%%%%%%%%%%%%%%%%%%%%%%%%%%%%%%

\end{document}